\documentclass[journal]{IEEEtran}

\usepackage{braket,amsfonts,amsthm}
\usepackage{tikz}
\usetikzlibrary{positioning,chains,fit,shapes,calc}
\usepackage{array}
\usepackage{amssymb}
\usepackage[caption=false]{subfig}
\usepackage{pgfplots}

\usepackage{cleveref}
\newtheorem{definition}{Definition} 
\newtheorem{theorem}{Theorem} 
\newtheorem{lemma}{Lemma} 
\newtheorem{corollary}{Corollary} 
\newtheorem{proposition}{Proposition} 
\newtheorem{rem}{Remark} 
\usepackage{algorithm, algorithmic}

\usepackage{algorithmic}

\usepackage{graphicx,epstopdf}

\Crefname{ALC@unique}{Line}{Lines}

\usepackage{amsopn}
\usepackage{multirow}
\usepackage{xspace}
\usepackage{bold-extra}
\usepackage[most]{tcolorbox}

\colorlet{texcscolor}{blue!50!black}
\colorlet{texemcolor}{red!70!black}
\colorlet{texpreamble}{red!70!black}
\colorlet{codebackground}{black!25!white!25}

\lstdefinestyle{siamlatex}{%
  style=tcblatex,
  texcsstyle=*\color{texcscolor},
  texcsstyle=[2]\color{texemcolor},
  keywordstyle=[2]\color{texemcolor},
  moretexcs={cref,Cref,maketitle,mathcal,text,headers,email,url},
}

\tcbset{%
  colframe=black!75!white!75,
  coltitle=white,
  colback=codebackground, 
  colbacklower=white, 
  fonttitle=\bfseries,
  arc=0pt,outer arc=0pt,
  top=1pt,bottom=1pt,left=1mm,right=1mm,middle=1mm,boxsep=1mm,
  leftrule=0.3mm,rightrule=0.3mm,toprule=0.3mm,bottomrule=0.3mm,
  listing options={style=siamlatex}
}

\newtcblisting[use counter=example]{example}[2][]{%
  title={Example~\thetcbcounter: #2},#1}

\newtcbinputlisting[use counter=example]{\examplefile}[3][]{%
  title={Example~\thetcbcounter: #2},listing file={#3},#1}

\DeclareTotalTCBox{\code}{ v O{} }
{ 
  fontupper=\ttfamily\color{black},
  nobeforeafter,
  tcbox raise base,
  colback=codebackground,colframe=white,
  top=0pt,bottom=0pt,left=0mm,right=0mm,
  leftrule=0pt,rightrule=0pt,toprule=0mm,bottomrule=0mm,
  boxsep=0.5mm,
  #2}{#1}

\patchcmd\newpage{\vfil}{}{}{}
\begin{document}
	\title{Stochastic Non-preemptive Co-flow Scheduling with Time-Indexed Relaxation}
	\author{Ruijiu Mao, Vaneet Aggarwal, and Mung Chiang \thanks{The authors are with Purdue University, West Lafayette IN 47907 (email: \{mao95, vaneet, chiang\}@purdue.edu). This paper will be presented in part at the IEEE Infocom Workshop on Big Data in Cloud Performance (DCPerf), April 2018 \cite{rmao2018infwrk}. }}
	
\maketitle



\begin{abstract}
Co-flows model a modern scheduling setting that is commonly found in a variety of applications in distributed and cloud computing. A stochastic co-flow task contains a set of parallel flows with randomly distributed sizes. Further, many applications require non-preemptive scheduling of co-flow tasks. This paper gives an approximation algorithm for stochastic non-preemptive co-flow scheduling. The proposed approach uses a time-indexed linear relaxation, and uses its solution to come up with a feasible schedule. This algorithm is shown to achieve a competitive ratio of $(2\log{m}+1)(1+\sqrt{m}\Delta)(1+m{\Delta}){(3+\Delta)}/{2}$ for zero-release times, and $(2\log{m}+1)(1+\sqrt{m}\Delta)(1+m\Delta)(2+\Delta)$ for general release times, where $\Delta$ represents the upper bound of squared coefficient of variation of processing times, and $m$ is the number of servers.


\end{abstract}



\section{Introduction}
\label{sec:intro}

Computation frameworks  such as MapReduce \cite{Dean:2008:MSD:1327452.1327492}, Hadoop \cite{shvachko2010hadoop}, Spark \cite{spark}, and Google Dataflow \cite{google} are growing at an unprecedented speed. These frameworks enable users to offload computation to the cloud. In order for cloud service provider to maintain efficient services, they need to schedule the different jobs so as to minimize the completion time of the jobs. One of the key challenges in cloud computing is the data transmission across machines \cite{chowdhury2011managing}, which typically happens during the shuffle phase in the MapReduce based computations. In this paper, we will provide algorithms to reduce this communication time for different flows required in each of computational jobs.

Scheduling for shuffle phase is studied in the literature as co-flow scheduling \cite{Chowdhury:2012:CNA:2390231.2390237}. In this framework, a flow consists of data transfer between two servers. A co-flow task consists of multiple flows. A co-flow task is complete when all these flows are complete. Co-flow scheduling problem aims to schedule multiple co-flow tasks so that the weighted completion time of the co-flow tasks is minimized.  In most realistic big data computing jobs, the size of the flows to be transfered is not deterministic.  The authors of \cite{lawler1989sequencing} provide an overview of various scheduling problems with random parameters. For instance, processing times can be regarded as independent random variables drawn from given probability distributions.   Further, many scenarios do not allow for stopping a transfer once started \cite{kavi2001scheduled} and thus non-preemptive scheduling strategies are important. The key reasons for practicality of non-preemption include additional signaling overhead, flow switching latency, packet drops, and limitations on hardware. Thus, this paper considers non-preemptive stochastic scheduling of multiple co-flow tasks.




Aiming to reduce weighted completion time of tasks, this paper proposes a co-flow scheduling algorithm to order each constituent flow of each co-flow task with a random data size on each link. The non-preemptive constraint implies that  each flow occupies whole capacity of its source and sink machines, namely one unit per time slot, until the completion of the flow. Stochastic non-preemptive co-flow model provides flexibility and efficiency based on parallelism: constituent flows from several co-flows might be processed at same time. The problem even with deterministic flow sizes is NP-hard. The authors of \cite{yu2016non} considered a non-preemptive co-flow problem, where the different links were assigned bandwidth and thus the links were independent. In the standard co-flow problem \cite{Qiu:2015:MTW:2755573.2755592,shafiee2017scheduling, DBLP:journals/corr/ImP17} that is also considered in this paper, there are flow constraints at the source and the sink ends.



We note that the problem of scheduling preemptive co-flows with deterministic flow time has been considered in \cite{Qiu:2015:MTW:2755573.2755592,shafiee2017scheduling, DBLP:journals/corr/ImP17}, where $O(1)$-approximation algorithms are provided for minimizing the weighted completion time. However, there are no corresponding results for scheduling non-preemptive co-flows. This is the first paper, to the best of our knowledge, that provides approximation guarantees for the non-preemptive co-flow scheduling problem. Let $\Delta$ represents the upper bound of squared coefficient of variation of processing times and $m$ is the number of servers.
Then, the algorithm proposed in this paper is an $(2\log{m}+1)(1+\sqrt{m}\Delta)(1+m{\Delta}){(3+\Delta)}/{2}$ approximation algorithm for zero release time, and $(2\log{m}+1)(1+\sqrt{m}\Delta)(1+m{\Delta})(2+\Delta)$ approximation algorithm for general release times.

The proposed algorithm uses a time-slotted model, where the processing times are integer multiple of the length of the time slot. A linear programming (LP) based relaxation algorithm is formulated, that has variable of probability distribution of start of co-flow on each link, thus providing the average completion time of each co-flow. Since this is only a relaxation, the schedule may not be feasible. Based on weighted shortest expected processing time algorithm (WSEPT), an optimal rule for single-machine scheduling with stochastic processing time \cite{rothkopf1966scheduling}, we generate tentative start time for every constituent flow suggested by the LP solution, and group the flows by their tentative start time. Further,  a  grouping of coflows (originally used for input-queued switches \cite{keslassy2003guaranteed}) is used which provides groups of co-flows which could be scheduled simultaneously since they have no conflicts.  Scheduling  these groups in order gives the proposed algorithm.




The main contributions of the paper can be summarized as follows.

\begin{enumerate}
	\item This paper is the first paper, to the best of our knowledge, that gives approximation guarantees for stochastic non-preemptive co-flow scheduling. The results have been provided both with zero release times, and general release times.
	\item As a special case of stochastic, $\Delta=0$ gives the results for deterministic non-preemptive co-flow scheduling. These are also the first approximation results, to the best of our knowledge, for this case, where the approximation guarantees are  $3/2(2\log{m}+1)$ approximation algorithm for zero release times, and $2(2\log{m}+1)$ approximation algorithm for general release times.
\end{enumerate}

The rest of this paper is organized as follows. Section \ref{sec:related} introduces the related work of co-flow scheduling and input-queued switches. Section \ref{sec:model} presents the formulation of the problem. In Section \ref{sec:approach}, the proposed algorithm is provided. Section  \ref{sec:guarantees}  proves the approximation bounds of the proposed algorithm. Section \ref{sec:release} extends the algorithm and the results to general release times.  Finally, Section \ref{sec:con} concludes the paper. 


\section{Related Work}
\label{sec:related}

In this section, we will describe the related work for this paper in three categories of co-flow scheduling, input-queued switch, and stochastic scheduling on parallel machines.

\subsection{Co-flow Scheduling}
The concept of co-flow was first proposed in 2012 by Chowdhury \cite{Chowdhury:2012:CNA:2390231.2390237}. Further, the authors of \cite{Chowdhury:2014:ECS:2740070.2626315}  proposed an efficient  implementation of co-flow scheduling. However, these works did not provide approximation guarantees for the proposed algorithm.  


Polynomial-time  approximation algorithms have been proposed for deterministic preemptive co-flow scheduling in \cite{Qiu:2015:MTW:2755573.2755592,shafiee2017scheduling,DBLP:journals/corr/ImP17}. The authors of \cite{Qiu:2015:MTW:2755573.2755592} use a relaxed linear program followed by a Birkhoff-von-Neumann (BV) decomposition to schedule flows. In contrast, this paper considers non-preemptive scheduling, and has stochastic task sizes. 


Non-preemptive deterministic co-flow scheduling has also been studied in \cite{yu2016non}. However, the authors assumed fixed bandwidth links between every pair of servers. In contrast, this paper considers bandwidth constraints at source and sink which is typical for co-flow scheduling literature.


\subsection{Input-queued switch}
To decrease the frequency of matching computation for crossbar configuration in high-speed core routers, frame-based scheduling has being widely studied \cite{mckeown1999achieving,mckeown1999islip}. One of the essential steps of frame-based scheduling is the computation of a list of input/output pair. The pairing step is similar to co-flow scheduling in the sense of grouping constituent flows to be processed in the same time slot. The listing step is similar to co-flow scheduling in the sense of sequencing the groups flows at each time slot. The authors of \cite{keslassy2003guaranteed} introduced Greedy Low Jitter Decomposition (GLJD) Algorithm to solve list pairing scheduling problem. The GLJD algorithm can be seen in non-preemptive scheduling as an algorithm that has a similar goal as the BV decomposition for preemptive scheduling in \cite{Qiu:2015:MTW:2755573.2755592}. 


\subsection{Stochastic scheduling on parallel machines}
For  scheduling of jobs, the size of tasks is unknown apriori. Thus, introducing randomness in the task sizes is a natural abstraction. In 1966, the authors of \cite{rothkopf1966scheduling} proved that WSEPT rule (weighted shortest expected processing time first) is optimal for minimizing weighted completion time of jobs on a single server. Based on WSEPT rule, the authors of \cite{skutella2016unrelated} studied the stochastic variant of unrelated parallel machine scheduling. This approach uses a time-indexed linear programming relaxation for stochastic machine scheduling, whose ideas have been used in this paper to formulate a time-indexed linear relaxation for stochastic co-flow scheduling.


\section{Problem Formulation}
\label{sec:model}

In this section, we will describe the problem of Stochastic Non-preemptive Co-flow scheduling. Consider set of $m$ servers given as ${\cal M} = \{1, 2, \cdots, m\}$. A flow represents a data communication between a server $i\in {\cal M}$ and a server $j\in {\cal M}$. Each co-flow task is composed of multiple flows. Let there be $N$ co-flow tasks, where co-flow task $k$ is indexed by a set ${\cal T}_k$. This set ${\cal T}_k$ is composed of a set of flows represented by a set of $(i,j,k)$, where $i, j \in {\cal M}$, and $k\in \{1, \cdots, N\}$. For each such flow, the size of the flow is characterized by  $S_{i,j,k}$, which is the random variable indicating the size of flow (or data transfer) from $i$ to $j$. More, formally, a co-flow task is defined as follows. 

\begin{definition}
$k$-th co-flow task is defined as a set of flows, or ${\cal T}_k \subseteq \{(i,j, k): i, j \in {\cal M}, k\in \{1, \cdots, N\}\}$. Further,  $S_{i,j,k}$ is the random variable indicating the size of flow (or data transfer) from $i$ to $j$.	
\end{definition}

We note that the size of flows can be discretized to positive integers, while only loosing a factor $1 +\epsilon$ in the objective function value for any $\epsilon>0$, where the number of discrete levels are $O(1/\epsilon)$ following similar proof as in \cite[Lemma 1]{skutella2016unrelated}. Thus, we will assume that the random variable $S_{i,j,k}$ only takes non-negative integer values.

We next define the notion of co-flow scheduling. Co-flow scheduling problem is to schedule the different flows on each link $(i,j)$, where the different flows are given as $\cup_{k=1}^N \{(i,j,k): \Pr(S_{i,j,k}=0)<1\}$. Non-preemptive scheduling implies that once a task $(i,j,k) \in {\cal T}_k$ is started, it will be processed till completion. By stochastic, we mean that the $S_{i,j,k}$ is the random variable, whose cumulative distribution function is known. We assume that the probability that $S_{i,j,k}$ is at least $t$ be $p_{i,j,k,t}$, or $p_{i,j,k,t} = \Pr(S_{i,j,k}\ge t)$. 



We assume a time-slotted model. We partition time into time slots $(t\in\{0, 1, \cdots\})$. For example, $t=0$ is the first time slot with unit time length. Further, we assume that each source port can only send one unit of data per time slot and every sink port can only receive one unit of data per time slot.




Let $w_k, k\in \{1, \cdots, N\}$ be the weights of the different co-flows. Let the expected completion time of a co-flow $(i,j,k)\in {\cal T}_k$ be $C_{i,j,k}$. Further, the expected completion time of a co-flow task ${\cal T}_k$ is given as $\max_{(i,j,k)\in {\cal T}_k} C_{i,j,k}$. Based on these, the Stochastic Non-Preemptive Co-flow scheduling is defined as follows.

\begin{definition}
Stochastic Non-Preemptive Co-flow scheduling wishes to find the order of scheduling non-preemptive co-flows on each link with stochastic processing times, so as to minimize $\sum_{k=1}^N w_k C_k$.  
\end{definition}



\begin{figure}[hbtp]
\definecolor{myblue}{RGB}{80,80,160}
\definecolor{mygreen}{RGB}{80,160,80}

\begin{tikzpicture}[thick,
every node/.style={draw,circle},
fsnode/.style={fill=myblue},
ssnode/.style={fill=mygreen},
every fit/.style={ellipse,draw,inner sep=-2pt,text width=2cm},
->,shorten >= 3pt,shorten <= 3pt
]

\begin{scope}[start chain=going below,node distance=7mm]
\foreach \i in {1,2,...,5}
\node[fsnode,on chain] (f\i) [label=left: \i] {};
\end{scope}

\begin{scope}[xshift=6cm,yshift=0cm,start chain=going below,node distance=7mm]
\foreach \i in {1,2,...,5}
\node[ssnode,on chain] (s\i) [label=right: \i] {};
\end{scope}

\node [myblue,fit=(f1) (f5),label=above: Source Ports] {};
\node [mygreen,fit=(s1) (s5),label=above: Sink Ports] {};
\draw (f1) -- (s1) node[draw=none,fill=none,pos=0.4]{$1,2$};
\draw (f2) -- (s1)node[draw=none,fill=none,pos=0.4]{$3$};
\draw (f2) -- (s2)node[draw=none,fill=none,pos=0.4]{$2,3$};
\draw (f2) -- (s4)node[draw=none,fill=none,pos=0.4]{$1$};
\draw (f4) -- (s3)node[draw=none,fill=none,pos=0.2]{$1$};
\draw (f3) -- (s4)node[draw=none,fill=none,pos=0.2]{$3$};
\draw (f5) -- (s5)node[draw=none,fill=none,pos=0.4]{$2$};
\draw (f5) -- (s1)node[draw=none,fill=none,pos=0.6]{$3$};
\end{tikzpicture}
\caption{An example to demonstrate the co-flows on each link. The number on the links represents the task numbers that has a flow on that link.}\label{ex_figure}
\end{figure}
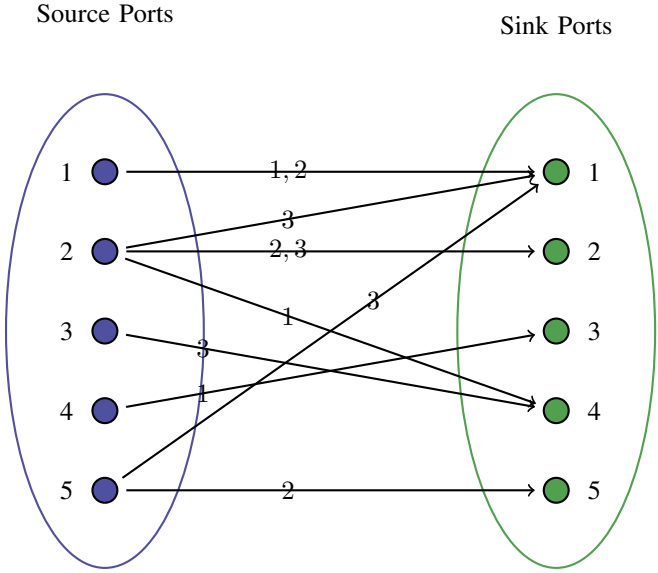

As an example, consider $m=5$ servers as depicted in Figure \ref{ex_figure}. We wish to schedule three co-flows given as follows.

\begin{align*}
	{\cal T}_1&=\{(1, 1, 1),(2, 4, 1),(4, 3, 1)\},\\
	{\cal T}_2&=\{(1, 1, 2),(2, 2, 2),(5, 5, 2)\},\\
	{\cal T}_3&=\{(2, 1, 3),(2, 2, 3),(3, 4, 3),(5, 1, 3)\}.
\end{align*}
All flows have stochastic sizes $S_{i,j,k}$, but we know their distributions. By capacity constraints, flows $(1,1,1)$, $(1,1,2)$, $(2,1,3)$, and $(5,1,3)$ can not be processed simultaneously since they share the same sink port. Similarly, $(2,1,3)$, $(2,2,2)$, $(2,2,3)$, $(2,4,1)$ can not be processed simultaneously since they share the same source port.







\section{Proposed Algorithm}
\label{sec:approach}

The proposed algorithm uses a linear programming (LP) relaxation of the co-flow scheduling problem. Since the solution of the relaxed problem may not in general be feasible, an algorithm using the solution of the relaxed problem that gives a feasible schedule will be provided. Guarantees that the proposed algorithm is approximately optimal will be derived in Section \ref{sec:guarantees}. 

We first derive a LP relaxation of the stochastic non-preemptive co-flow scheduling problem. Let $y_{i,j,k,t}$ be the indicator that the flow $(i,j,k)$ will be started processing at time slot $t$. In the relaxation, we will relax the integer constraint so that $y_{i,j,k,t}$ represents the probability that flow $(i,j,k)$ will be started processing at time slot $t$. Further, let the optimization problem variable $C_k$ represent the expected completion time of $k$-th co-flow. Then, the LP relaxed problem to minimize the weighted expected completion time of the co-flows can be formulated as follows. 

\begin{align}
& \min &&\sum_{k=1}^N w_kC_k \label{1}\\
& \rm{s.t.} && \sum_{t=0}^{\infty}y_{i,j,k,t}=1 \label{2}\\
& && \qquad \forall i,j\in\mathcal{M},\quad\forall k\in\{1,\cdots,N\};\nonumber\\
& &&\sum_{j\in\mathcal{M}}\sum_{k=1}^N\sum_{t=0}^s y_{i,j,k,t}p_{i,j,k,s-t}\leq 1 \label{3}\\
& &&\qquad\forall i\in\mathcal{M},\quad s\in\{0,1,\cdots\};\nonumber\\
& &&\sum_{i\in\mathcal{M}}\sum_{k=1}^N\sum_{t=0}^s y_{i,j,k,t}p_{i,j,k,s-t}\leq 1\label{4}\\
& &&\qquad\forall j\in\mathcal{M},\quad s\in\{0,1,\cdots\};\nonumber\\
& &&C_k\geq \sum_{t=0}^\infty y_{i,j,k,t}(t+\mathbb{E}[S_{i,j,k}])\label{5}\\
& &&\qquad \forall i,j\in\mathcal{M},\quad  k\in\{1,\cdots,N\};\nonumber\\
& && y_{i,j,k,t}\geq 0 \label{6}\\
& &&\qquad \forall i,j\in\mathcal{M},\quad\forall k\in\{1,\cdots,N\},\quad t\in\{0,1,\cdots\}.\nonumber
\end{align}


Constraint \eqref{2} says that every constituent flow $(i,j,k)$ will be assigned some time. Constraints \eqref{3} and \eqref{4} are matching constraints, where on an average, at most one unit of data leaves the source or enters the sink. Constraint \eqref{5} says the expected completion time of a co-flow task is at least the maximum among all of its constituent flows' expected completion time. Constraint \eqref{6} says the probability that every constituent flow $(i,j,k)$ starts at time $t$ $(t\in\{0,1,\cdots\})$ is non-negative.

We denote the optimal value of $C_k$ from the LP relaxation problem as $C_k^{LP}$,  the optimal $y_{i,j,k,t}$ as $y_{i,j,k,t}^{LP}$, and the expected completion time of flow $(i,j,k)$  as $C_{i,j,k}^{LP} = \sum_{t=0}^\infty y_{i,j,k,t}^{LP}(t+\mathbb{E}[S_{i,j,k}])$

We first note that even though $t$ is being summed till $\infty$, it only needs to be summed till the sum of maximum flow sizes on each link.  In \cref{lma0}, we will show that the infinite time-indexed LP relaxation can be approximated by a finite time-indexed relaxation.  Further, we see that the constraints are necessary for the co-flow scheduling problem. Thus, the weighted expected completion time as the optimal objective of the LP relaxation is less than or equal to the optimal expected weighted completion time of the co-flows. In other words, 

\begin{equation}
\sum_{k=1}^N w_kC_k^{LP} \le \sum_{k=1}^N w_kC_k^*, \label{bnd_opt}
\end{equation}
where $\{C_1^*, \cdots, C_N^*\}$ are the expected completion times of the optimal co-flow scheduling.

We give a definition of pseudo-permutation matrix, which is used in the following algorithm.
\begin{definition}
A pseudo-permutation matrix is a square binary matrix that has at most one entry of $1$ in each row and each column and $0$s elsewhere. 
\end{definition}

We will now describe the proposed algorithm, Non-Preemptive Stochastic Co-flow Scheduling (NPSCS), which is summarized in Algorithm \ref{algo:npscs}

\begin{algorithm}[htbp]
	\begin{algorithmic}[1]
		\STATE {\bf Input:}  $p_{i,j,k,t}$, $\mathbb{E}[S_{i,j,k}]$, $(i,j\in\mathcal{M},k\in\{1,\cdots,N\})$. 
		\STATE {\bf Output:} A list of perfect matching co-flows to be scheduled in turn, $\Gamma$.
		\STATE Solve LP problem \eqref{1}-\eqref{6}, and obtain $y_{i,j,k,t}^{LP}$ for all $(i,j\in\mathcal{M},k\in\{1,\cdots,N\})$.
		\FOR {$i\in\{1,\cdots,m\}$ }
		\FOR{ $j\in\{1,\cdots,m\}$ } 
		\FOR {$k\in\{1,\cdots,N\}$}
		\STATE Choose $t \in \{0,1, \cdots \}$ such that the probability mass function of $t$ is $\Pr(t=a)=y_{i,j,k,a}^{LP}$. 
		\STATE Choose $r\in\{0,1,\cdots, \}$ such that the probability mass function of $r$ is  $\Pr(r=b) = p_{i,j,k,b}/\mathbb{E}[S_{i,j,k}]$.
		\STATE Compute $t(i,j,k)=t+r$ as the tentative start time of flow $(i,j,k)$.
		\ENDFOR
		\ENDFOR
		\ENDFOR
		\FOR{ $s\in\{0,1,\cdots\}$}
		\STATE Create a matrix ${\textbf D(s)}\in\mathbb{R}^{m\times m}$. Denote ${\textbf D(s)}_{i,j}$ as the $(i,j)$ entry of ${\bf D}(s)$. Initialize each entry of ${\bf D}(s)$ as zero.
		\STATE Create a set $\mathcal{J}(s)=\{~\}$
		\FOR{ $i\in\{1,\cdots,m\}$ }
		\FOR{ $j\in\{1,\cdots,m\}$ } 
		\FOR{ $k\in\{1,\cdots,N\}$}
		\IF {$t(i,j,k)==s$}
		\STATE ${\textbf D(s)}_{i,j} = {\textbf D(s)}_{i,j}+\mathbb{E}[S_{i,j,k}]$
		\STATE ${\mathcal{J}(s)}=\mathcal{J}(s)\bigcup (i,j,k)$
		\ENDIF
		\ENDFOR
		\ENDFOR
		\ENDFOR
		\STATE Input matrix ${\bf D}(s)$ to \cref{algo:GLJD} to obtain a set of pseudo-permutation matrices $\{{\bf X}(s)^1,\cdots,{ \bf X}(s)^{l_s}\}$, where $l_s$ is the number of resulting pseudo-permutation matrices from \cref{algo:GLJD}
		\FOR {$l\in\{1,\cdots, l_s\}$}
		\STATE {${\mathcal{I}(s)^l}=\{~\}$}
		\FOR {$i\in\{1,\cdots,m\}$ }
		\FOR{ $j\in\{1,\cdots,m\}$ } 
		\IF {${\textbf X(s)^l}_{i,j}==1$}
		\STATE {${\mathcal{I}(s)^l}=\Big{(}\bigcup_{k=1}^N (i,j,k)\Big{)}\bigcap {\mathcal{J}(s)}$}		
		\ENDIF
		\ENDFOR
		\ENDFOR
		\ENDFOR
		\ENDFOR
		\STATE $\Gamma = [~]$
		\FOR{ $s\in\{0,1,\cdots\}$}
		\FOR{ $l\in\{1,\cdots,l_s\}$}
		\STATE $\Gamma = [\Gamma, \mathcal{I}(s)^l]$
		\ENDFOR
		\ENDFOR
	\end{algorithmic}
	\caption{Non-Preemptive Stochastic Co-flow Scheduling (NPSCS) }\label{algo:npscs}
\end{algorithm}

For non-preemptive scheduling, the problem reduces to deciding the start time of all the constituent flows. The proposed scheduling algorithm, NPSCS,  summarized in  \cref{algo:npscs},  consists of 4 steps. 

In Step 1 (line 3), we solve the relaxed LP problem \eqref{1}-\eqref{6} and get the optimal probability $y_{i,j,k,t}^{LP}$ for each constituent flow $(i,j,k)$ starting at each time slot $t$. Although we can solve the LP and get optimal solution $\sum_{k\in\{1,\cdots,N\}}w_kC_k^{LP}$, the solution is not in general a feasible schedule.  However, $\sum_{k=1}^Nw_kC_k^{LP}$ is a lower bound for the optimal scheduling, and $y_{i,j,k,t}^{LP}$ can provide insights on scheduling the start time of flow $(i,j,k)$.

In Step 2 (lines 4-12) and Step 3 (lines 13-37), we show how to turn the time-indexed LP relaxation to a feasible schedule. In Step 2, we generate $t$ with probability mass function $y_{i,j,k,t}^{LP}$ and $r\in\{0,1,\cdots\}$ with probability mass function $p_{i,j,k,r}/\mathbb{E}[S_{i,j,k}]$. We define the tentative start time $t(i,j,k)$ for each flow $(i,j,k)$ by summing up $t$ and $r$. 

In Step 3 (line 13-27), we group all the flows with same tentative start time $s$ as $\mathcal{J}(s)$. We build a matrix ${\bf D}(s)\in \mathbb{R}^{m\times m}$ with only non-negative elements. The intersection of sets $\{(i,j,k)|k\in\{1,\cdots,N\}$ and $\mathcal{J}(s)$ is the set of flows with tentative start time $s$ and transfer data from server $i$ to server $j$. We sum the expectation size of flows in the intersection set up and set $\big{(}{\bf D}(s)\big{)}_{i,j}$ to be the value of this sum. We want to process all flows within $\mathcal{J}(s)$ simultaneously, but the capacity constraints of each server do not allow flows with same source server or same sink server to be processed on one time slot. Only flows having no interference on source or sink servers can be processed at the same time slot. As a result, we use Greedy Low Jitter Decomposition algorithm (GLJD) given in  \cref{algo:GLJD} to get a set of pseudo-permutation matrices $\{{\bf X}(s)^1,\cdots,{\bf X}(s)^{l_s}\}$, where $l_s$ is the number of resulting pseudo-permutation matrices.  For any ${\bf X}(s)^{l}$ where $l\in\{1,\cdots,l_s\}$, we denote all flows in $\mathcal{J}(s)$ from server $i$ to server $j$ for which $\big{(}{\bf X}(s)^{l}\big{)}_{i,j}$ is 1 as $\mathcal{I}(s)^l$. 

GLJD has originally been proposed in \cite{keslassy2003guaranteed} to  enable traffic scheduling with low-jitter guarantees.  GLJD returns a set of pseudo-permutation matrices, whose structures provides feasible co-flow scheduling satisfying capacity constraints.
For $s\in\{0,1,\cdots\}$, GLJD algorithm (\cref{algo:GLJD}) first sort the expectation size of all flows from $\mathcal{J}(s)$ in non-increasing order to create a list $\mathcal{L}$. We record the source server and sink server locations of the $n$-th largest expected size as $\rho(n)$ and $\kappa{n}$. We greedily pick elements from the top to the bottom of the list $\mathcal{L}$ as long as the elements do not share same source server or sink server with the help of $\rho(\cdot)$ and $\kappa(\cdot)$. Once we finish searching the list, we create a pseudo-permutation matrix $X_s^{l}$ ($l$ is the number of times we search the list from the beginning) with entry $(i,j)$ equal to one if and only if at least one element picked from $\mathcal{L}$ has $\rho(\cdot)=i$ and $\kappa(\cdot)=j$ ($\forall i,j\in\{1,\cdots,m\}$).  Next, we delete the corresponding elements from list $\mathcal{L}$ and start constructing for the next pseudo-permutation matrix until the list $\mathcal{L}$ becomes empty. 

In Step 4 (lines 38-43, \cref{algo:npscs}), the NPSCS algorithm schedules flows in the order of $\Gamma$, which is a concatenation of $\mathcal{I}(s)^l$,  for all $s\in\{0,1,\cdots\}$, and $l\in\{1,\cdots,l_s\}$ as can be seen in \cref{algo:npscs}.

\begin{algorithm}[htbp]
	\begin{algorithmic}[1]
		\STATE {\bf Input:}  A matrix ${\bf D}(s)\in\mathbb{R}^{m\times m}$ with only non-negative elements 
		\STATE {\bf Output:} A set of pseudo-permutation matrices $\{{\bf X}(s)^1,\cdots,{\bf X}(s)^{l_s}\}$, where $l_s$ is the number of resulting matrices in the set
		\STATE Create a list $\mathcal{L}$ of all non-zero entries in ${\bf D}(s)$ by non-increasing order of their values
		\STATE $a_{\mathcal{L}}=|\mathcal{L}|$ ($|\mathcal{L}|$ is the number of elements in list $\mathcal{L}$)
		\STATE {Set $\mathcal{M}=\{(i,j)|{\rm for~all~} {\bf D}(s)_{i,j}>0, i\in\{1,\cdots,m\},j\in\{1,\cdots,m\}\}$}
		\FOR {$n\in\{1,\cdots,a_{\mathcal{L}}\}$}
		\FOR {$(i,j)\in\mathcal{M}$}
		\IF {$\mathcal{L}(n)=={\bf D}(s)_{i,j}$}
		\STATE {set $\rho(n)=i$ and $\kappa(n)=j$}
		\STATE {Remove $(i,j)$ from $\mathcal{M}$}
		\ENDIF
		\ENDFOR
		\STATE {$n = n+1$}
		\ENDFOR
		\STATE Initialize $l=1$
		\WHILE{$\mathcal{L}\neq \emptyset$}
		\STATE Set ${\bf X}(s)^l={\bf{0}^{m\times m}}$
		\STATE Set $a=1$
		\STATE Set $\mathcal{C}[k]=0$ $\forall k\in\{1,\cdots,m\}$
		\WHILE {$a\leq a_{\mathcal{L}}$ }
		\WHILE {$\mathcal{C}[\rho(a)]==0$ and $\mathcal{C}[\kappa(a)]==0$}
		\STATE $\big{(}{\bf X}(s)^l\big{)}_{\rho(a),\kappa(a)} = 1$
		\STATE $\mathcal{C}[\rho(a)]=\mathcal{C}[\kappa(a)]=1$
		\STATE Eliminate entry $a$ from list $\mathcal{L}$
		\ENDWHILE
		\STATE $a = a+1$
		\ENDWHILE 
		\STATE $l =  l+1$
		\ENDWHILE
	\end{algorithmic}
	\caption{Greedy Low Jitter decomposition (GLJD): }\label{algo:GLJD}
\end{algorithm}

For a single $s$, we now provide an example to explain step 3 (line 13-27) and step 4 (line 38-43).  Suppose we have $4$ servers and $3$ co-flows. Further, suppose that the  flows given in Table \ref{tbl:ex} have $s$  as their tentative start time.
\begin{table}
\begin{center}
\begin{tabular}{ |c | c | c | c | c | c |}
\hline $(i,j,k)$ &  $\mathbb{E}[S_{i,j,k}]$ & $(i,j,k)$ &  $\mathbb{E}[S_{i,j,k}]$ & $(i,j,k)$ &  $\mathbb{E}[S_{i,j,k}]$ \\
\hline $(1,1,1)$ & $0.38$ & $(2,2,2)$  & $0.24$   & $(1,3,2)$ & $0.22$ \\
\hline $(2,1,3)$ & $0.11$ &  $(3,2,1)$& $0.19$ & $(2,3,1)$& $0.20$ \\
\hline $(4,1,1)$  & $0.20$& $(3,2,2)$  &  $0.31$ & $(2,3,2)$ & $0.20$ \\
\hline $(4,1,3)$ &  $0.31$&  $(3,2,3)$   &  $0.03$   & $(2,3,3)$& $0.20$  \\
\hline $(1,4,1)$ &  $0.40$ &  $(4,2,2)$ &  $0.23$ & $(3,3,3)$ & $0.14$ \\
\hline $(2,4,1)$  &  $0.05$   &   $(4,4,1)$&  $0.22$ & $(4,3,3)$ & $0.04$\\
\hline  $(3,4,2)$ & $0.33$ & - & - &-& -\\
\hline
\end{tabular}  \label{tbl:ex}\caption{The flows with tentative start time $s$ in the example}
\end{center}
\end{table}
The flows in Table \ref{tbl:ex} form the matrix $\bf{D(s)}$, which is given  as 
\begin{equation*}\bf{D(s)}=
\begin{bmatrix}
0.38 & 0 & 0.22 & 0.40\\
0.11 & 0.24 & 0.60 & 0.05\\
0 & 0.53 & 0.14 & 0.33\\
0.51 & 0.23 & 0.04 & 0.22
\end{bmatrix},
\end{equation*}
and $\mathcal{J}(s)$ is given as 
\begin{eqnarray*}
\mathcal{J}(s)&=&\{(1,1,1),(2,1,3), (4,1,1), (4,1,3), (2,2,2), \\
&&(3,2,1), (3,2,2), (3,2,3), (4,2,2), (1,3,2),\\
&& (2,3,1),(2,3,2), (2,3,3), (3,3,3), (4,3,3),\\
&& (1,4,1), (2,4,1), (3,4,2), (4,4,1)\}.
\end{eqnarray*}

From {\bf Algorithm} \ref{algo:GLJD}, we have the set of psedo-permutation matrices:
\begin{align*}\bf{X(s)^1}&=
&\begin{bmatrix}
0 & 0 & 0& 1\\
0& 0& 1 & 0\\
0 & 1 & 0 & 0\\
1 & 0 & 0 & 0
\end{bmatrix}&,
\bf{X(s)^2}=
\begin{bmatrix}
1& 0 & 0& 0\\
0& 1& 0 & 0\\
0 & 0 & 0 & 1\\
0 & 0 & 1 & 0
\end{bmatrix},\\
\bf{X(s)^3}&=
&\begin{bmatrix}
0& 0 & 1& 0\\
1&0& 0 & 0\\
0 & 0 & 0 & 0\\
0 & 1 & 0 & 0
\end{bmatrix}&,
\bf{X(s)^4}=
\begin{bmatrix}
0& 0 & 0 & 0\\
0&0& 0 & 0\\
0 & 0 & 1 & 0\\
0 & 0 & 0 & 1
\end{bmatrix},\\
\bf{X(s)^5}&=&
\begin{bmatrix}
0& 0 & 0& 0\\
0&0& 0 & 1\\
0 & 0 & 0 & 0\\
0 & 0 & 0 & 0
\end{bmatrix}&.
\end{align*}

Thus, \cref{algo:GLJD} gives  $l_s=5$. For $l=1$, we take all the flows in $\mathcal{J}(s)$ having in the form of $(1,4,\cdot), (2,3,\cdot), (3,2,\cdot), (4,1,\cdot)$ to $\mathcal{I}(s)^1$. Therefore, 
\begin{eqnarray*}
\mathcal{I}(s)^1&=&\{(1,4,1),(2,3,1),(2,3,2),(2,3,3),(3,2,1),\\
&&(3,2,2),(3,2,3),(4,1,1),(4,1,3)\}.
\end{eqnarray*}

For $l=2$, we take all the flows in $\mathcal{J}(s)$ having in the form of $(1,1,\cdot), (2,2,\cdot), (3,4,\cdot), (4,3,\cdot)$ to $\mathcal{I}(s)^2$. Therefore, 
\begin{eqnarray*}
\mathcal{I}(s)^2=\{(1,1,1),(2,2,2),(3,4,2),(4,3,3)\}.
\end{eqnarray*}

For $l=3$, we take all the flows in $\mathcal{J}(s)$ having in the form of $(1,3,\cdot), (2,1,\cdot), (4,2,\cdot)$ to $\mathcal{I}(s)^3$. Therefore, 
\begin{eqnarray*}
\mathcal{I}(s)^3=\{(1,3,2),(2,1,3),(4,2,2)\}.
\end{eqnarray*}

For $l=4$, we take all the flows in $\mathcal{J}(s)$ having in the form of $(3,3,\cdot), (4,4,\cdot)$ to $\mathcal{I}(s)^4$. Therefore, 
\begin{eqnarray*}
\mathcal{I}(s)^4=\{(3,4,2),(4,4,1)\}.
\end{eqnarray*}

For $l=5$, we take all the flows in $\mathcal{J}(s)$ having in the form of $(2,4,\cdot)$ to $\mathcal{I}(s)^5$. Therefore, 
\begin{eqnarray*}
\mathcal{I}(s)^5=\{(2,4,1)\}.
\end{eqnarray*}

We add $\mathcal{I}(s)^1,\mathcal{I}(s)^2,\mathcal{I}(s)^3,\mathcal{I}(s)^4,\mathcal{I}(s)^5$ to $\Gamma$ in order. When all flows with tentative start time less than $s$ are scheduled, we start  processing the flows in $\mathcal{I}(s)^1$ on different servers simultaneously. By processing, we mean that these flows are enqueued in the corresponding links thus forming a schedule for this task. We  start to process flows in $\mathcal{I}(s)^{i+1}$ after all the flows in $\mathcal{I}(s)^{i}$ are scheduled. Also, we  process flows in $\mathcal{I}(s+1)^{1}$ after all flows in $\mathcal{I}(s)^{l_s}$ are scheduled.

\section{Approximation Guarantee for NPSCS}
\label{sec:guarantees}

In this section, we will show that the proposed algorithm, NPSCS, is an approximation algorithm with bounded worst case approximation factor. First, we will show that even though the LP problem has unbounded time slots, we can upper bound the number of time slots to be a pseudo-polynomial in the input size. This will be followed by the approximation guarantee of the proposed algorithm. 

\subsection{Approximation by truncating the number of time slots}

Recall that the LP formulation \eqref{1}-\eqref{6} contains infinitely many variables and constraints since $t$ is summed till $\infty$. Inspired by Lemma 4 and Theorem 5 from \cite{skutella2016unrelated}, we show a pseudo-polynomial upper bound on the largest time index for $t$ in co-flow case. 

\begin{theorem}\label{lma0}
Suppose 
\begin{align*}
F_1&\triangleq mN \max_{k\in\{1,\cdots,N\}}\sum_{i\in\mathcal{M}}\sum_{j\in\mathcal{M}}\sum_{k=1}^N\mathbb{E}[S_{i,j,k}],\\
F_2&\triangleq 2mN\max_{i,j\in\mathcal{M},k\in\{1,\cdots,N\}}\mathbb{E}[S_{i,j,k}].
\end{align*}
Let $F\triangleq 2F_1+F_2$. Then, there is a set of optimal solutions of LP relaxation \eqref{1}-\eqref{6}: $\{y_{i,j,k,t}^{LP*}|i,j\in\mathcal{M},k\in\{1,\cdots,N\},t\in\{0,1,\cdots\}\}$ such that $y_{i,j,k,t}^{LP*}=0$ for $i,j\in\mathcal{M},k\in\{1,\cdots,N\},t>F$.
\end{theorem}
\begin{proof}
This result is akin to Lemma 4 in \cite{skutella2016unrelated}, which we prove here for completeness. 
For any $i\in\mathcal{M}$, we have the following:
\begin{eqnarray}\label{auxy}
&&\sum_{j\in\mathcal{M}}\sum_{k=1}^N\sum_{t\geq 2F_1}y_{i,j,k,t}\nonumber\\
&=&\sum_{j\in\mathcal{M}}\sum_{k=1}^N \mathbb{P}r({\rm time~to~start~}(i,j,k)\geq 2F_1)\nonumber\\
&\leq& \frac{\sum_{j\in\mathcal{M}}\sum_{k=1}^N \mathbb{E}[{\rm time~to~start~}(i,j,k)]}{2F_1}\nonumber\\
&\leq&\frac{F_1}{2F_1}\nonumber\\
&=&\frac{1}{2}.
\end{eqnarray}
The first equality follows from the definition of $y_{i,j,k,t}$. Recall that $y_{i,j,k,t}$ is the probability of starting flow $(i,j,k)$ at time $t$, $\sum_{t\geq 2F_1}y_{i,j,k,t}$ is the probability of starting flow $(i,j,k)$ no early than $2F_1$. The second inequality follows from Markov's inequality. The expected start time of an arbitrary flow starting at server is upper bounded by $\Big{(}\max_{k\in\{1,\cdots,N\}}r_{k}+\sum_{i\in\mathcal{M}}\sum_{j\in\mathcal{M}}\sum_{k=1}^N\mathbb{E}[S_{i,j,k}]\Big{)}$. Since there are at most $mN$ flows transferring data from server $i$. Therefore, the summation of expected start time of all flows from $i$ is bounded by $F_1$, leading to the third inequality.

Similarly, for any $j\in\mathcal{M}$, we have the following:
\begin{align*}
&\sum_{i\in\mathcal{M}}\sum_{k=1}^N\sum_{t\geq 2F_1}y_{i,j,k,t}\\
\leq&\sum_{i\in\mathcal{M}}\sum_{k=1}^N \mathbb{P}r({\rm time~to~start~}(i,j,k)\geq 2F_1)\\
\leq& \frac{\sum_{i\in\mathcal{M}}\sum_{k=1}^N \mathbb{E}[{\rm time~to~start~}(i,j,k)]}{2F_1}\\
\leq&\frac{1}{2}.
\end{align*}

Recall that $p_{i,j,k,r}$ is the probability that the $S_{i,j,k}$ is greater or equal to $r$. Based on the definition of $F_2$ and Markov's inequality, for any $r\geq F_2$, we have the following:
\begin{eqnarray}\label{auxp}
&&p_{i,j,k,r}=\mathbb{P}r[S_{i,j,k}\geq r]\nonumber\\
&\leq &\mathbb{P}r[S_{i,j,k}\geq 2mN\cdot\mathbb{E}[S_{i,j,k}]]\nonumber\\
&\leq & \frac{1}{2mN}.
\end{eqnarray}

Now we define a set of new LP solution $\{y_{i,j,k,r}^{LP*}|i,j\in\mathcal{M},k\in\{1,\cdots,N\}\}$:
\begin{equation*}
y_{i,j,k,r}^{LP*}\triangleq \left\{
\begin{aligned}
&y_{i,j,k,r}^{LP}\qquad\qquad &&r<F, \\
&\sum_{r'\geq F}y_{i,j,k,r'}^{LP} \qquad \qquad &&r=F, \\
&0\qquad\qquad &&r>F.
\end{aligned}
\right.
\end{equation*}

The new set of solutions will not get worse objective function \eqref{1} since the change will not make the original $C_k^{LP}$ larger based on \eqref{5}. Now we prove that the new set of solutions is feasible, satisfying constraints \eqref{2}-\eqref{6}. Satisfying of \eqref{2} and \eqref{6} can be seen in a straightforward fashion and are thus omitted. 

For \eqref{3}, if $s<F$, $y_{i,j,k,t}^{LP*}$ equals $y_{i,j,k,t}^{LP}$ for all $t\in\{0,\cdots,s\}$ by definition, \eqref{3} remains the same:
\begin{equation*}
\sum_{j\in\mathcal{M}}\sum_{k=1}^N\sum_{t=0}^sy_{i,j,k,t}^{LP*}p_{i,j,k,s-t}\leq 1.
\end{equation*}

If $s\geq F$,
\begin{eqnarray*}
&&\sum_{j\in\mathcal{M}}\sum_{k=1}^N\sum_{t=0}^sy_{i,j,k,t}^{LP*}p_{i,j,k,s-t}\\
&=&\sum_{j\in\mathcal{M}}\sum_{k=1}^N\sum_{t=0}^{2F_1-1}y_{i,j,k,t}^{LP}p_{i,j,k,s-t}\\
&& \qquad \qquad +\sum_{j\in\mathcal{M}}\sum_{k=1}^N\sum_{t=2F_1}^sy_{i,j,k,t}^{LP*}p_{i,j,k,s-t}\\
&\leq & \frac{1}{2mN}\sum_{j\in\mathcal{M}}\sum_{k=1}^N\sum_{t=0}^{2F_1-1}y_{i,j,k,t}^{LP}+\sum_{j\in\mathcal{M}}\sum_{k=1}^N\sum_{t\geq 2F_1}y_{i,j,k,t}^{LP}\\
&\leq & \frac{1}{2}+\frac{1}{2}=1.
\end{eqnarray*}
The first inequality follows from \eqref{auxp}, $p_{i,j,k,t}\leq 1$, and the definition of $y_{i,j,k,t}^{LP*}$. The second inequality follows from \eqref{auxy}, and $y_{i,j,k,t}\leq 1$.

Similarly, for \eqref{4}, if $s<F$, $y_{i,j,k,t}^{LP*}$ equals $y_{i,j,k,t}^{LP}$ for all $t\in\{0,\cdots,s\}$ by definition, \eqref{3} remains the same:
\begin{equation*}
\sum_{i\in\mathcal{M}}\sum_{k=1}^N\sum_{t=0}^sy_{i,j,k,t}^{LP*}p_{i,j,k,s-t}\leq 1.
\end{equation*}

If $s\geq F$,
\begin{eqnarray*}
&&\sum_{i\in\mathcal{M}}\sum_{k=1}^N\sum_{t=0}^sy_{i,j,k,t}^{LP*}p_{i,j,k,s-t}\\
&=&\sum_{i\in\mathcal{M}}\sum_{k=1}^N\sum_{t=0}^{2F_1-1}y_{i,j,k,t}^{LP}p_{i,j,k,s-t}\\
&& \qquad \qquad +\sum_{i\in\mathcal{M}}\sum_{k=1}^N\sum_{t=2F_1}^sy_{i,j,k,t}^{LP*}p_{i,j,k,s-t}\\
&\leq & \frac{1}{2mN}\sum_{i\in\mathcal{M}}\sum_{k=1}^N\sum_{t=0}^{2F_1-1}y_{i,j,k,t}^{LP}+\sum_{i\in\mathcal{M}}\sum_{k=1}^N\sum_{t\geq 2F_1}y_{i,j,k,t}^{LP}\\
&\leq & \frac{1}{2}+\frac{1}{2}\\
&=&1.
\end{eqnarray*}
The first inequality follows from \eqref{auxp}, $p_{i,j,k,t}\leq 1$, and the definition of $y_{i,j,k,t}^{LP*}$. The second inequality follows from \eqref{auxy}, and $y_{i,j,k,t}\leq 1$.

Thus, we have constructed a set of optimal solutions of LP relaxation \eqref{1}-\eqref{6}: $\{y_{i,j,k,t}^{LP*}\}$, satisfying $y_{i,j,k,t}^{LP*}=0$ for $i,j\in\mathcal{M},k\in\{1,\cdots,N\},t>F$.
\end{proof}
Based on this result, we have reduced the problem with infinite number of variables to a truncated time-indexed LP of pseudo-polynomial size.

We note that the intervals of geometrically increasing lengths can be chosen, as in \cite[Chapter 2.13]{skutella1998approximation}. In this discretization, the first interval $I_0 = [0,1]$, and the other intervals are $I_l = [(1+\epsilon)^{l-1}, (1+\epsilon)^{l}]$ for $l\ge 1$. This can lead to solving an approximation of time-indexed LP, albeit at an expense of losing a factor $1 + \epsilon$ in the objective function. 




\subsection{Approximation Result}
\label{sec:proofs}
In this section, we will prove that NPSCS is an approximation algorithm with a competitive ratio of $(2\log{m}+1)(1+\sqrt{m}\Delta)(1+m{\Delta}){(3+\Delta)}/{2}$, where $\Delta$ is the upper bound of $\mathbb{CV}[S_{i,j,k}]^2$ for all $i,j\in \mathcal{M}$, $k\in\{1,\cdots,N\}$. More formally, we have the following result. 

\begin{theorem} \label{thm5.1}
We aim to optimize the weighted completion time of $N$ co-flows on $m$ servers. The completion time of the $k$-th coflow tasks under scheduling algorithm NPSCS is at most $C_k^{LP}(2\log{m}+1)(1+\sqrt{m}\Delta)(1+m{\Delta}){(3+\Delta)}/{2}$, where $\Delta$ is the upper bound of $\mathbb{CV}[S_{i,j,k}]^2$ for all $i,j\in \mathcal{M}$, $k\in\{1,\cdots,N\}$. $\mathbb{CV}[S_{i,j,k}]^2\triangleq (\mathbb{E}[S_{i,j,k}^2]-\mathbb{E}[S_{i,j,k}]^2)/\mathbb{E}[S_{i,j,k}]^2$ is the squared coefficient of variation of $S_{i,j,k}$.
\end{theorem}

Before the proof of Theorem \ref{thm5.1}, we give a definition of efficient size for co-flow. We will later prove that each co-flow grouped by tentative start time has expected efficient size less or equal to one.

\begin{definition}[Efficient Size]\label{def}
Every stochastic co-flow can be represented as a matrix ${\bf D}\in\mathbb{R}^{m\times m}$ with its entry ${\bf D}_{ij}=\sum_{k=1}^n\mathbb{E}[S_{i,j,k}]$ representing the expected size of its constituent flow from server $i$ to server $j$. The efficient size of a stochastic co-flow is the maximum of the maximum column-sum and the maximum row-sum of its representative matrix. Namely, the efficient size of ${\bf D}$ is:
\begin{equation*}
\max\Big{\{}\max_i\sum_j \sum_k\mathbb{E}[S_{i,j,k}],\max_j\sum_i \sum_k\mathbb{E}[S_{i,j,k}]\Big{\}}.
\end{equation*}
\end{definition}
From Definition \ref{def}, we note that every $\mathcal{J}(s)$ has representative matrix ${\bf D}(s)$.

To prove Theorem \ref{thm5.1}, we use four lemmas. The first lemma proves that for all flows having the same tentative start time $s$, the stochastic co-flows grouped by same tentative start time $\mathcal{J}(s)$ has expected efficient size less or equal to $1$. If we assume the expected time to process a stochastic co-flow with efficient size less or equal to $1$ by Algorithm \ref{algo:GLJD} has upper bound $H$, namely the processing time for an arbitrary $\mathcal{J}(s)$ is bounded by $H$, then the summation of expected time scheduled before any flow with tentative start time $s$ has upper bound $(s+1/2)H$. 
\begin{lemma}\label{lma5.3}
Assume the expected time to process a stochastic co-flow with expected efficient size less or equal to $1$ by Algorithm \ref{algo:GLJD} has upper bound $H$. If we schedule all of the $N$ co-flows by policy \rm{NPSCS}, the total expected processing time before flow $(i,j,k)$ is at most $(t(i,j,k)+1/2)H$. Namely, the expected start time of flow $(i,j,k)$ is less or equal to $(t(i,j,k)+1/2)H$, where $t(i,j,k)$ is the tentative start time of flow $(i,j,k)$ in Algorithm \ref{algo:npscs}. In other words, the start time of set $\{\mathcal{I}(s)^1,\cdots,\mathcal{I}(s)^{l_s}\}$ is less or equal to $(s+1/2)H$. 
\end{lemma}
\begin{proof}
From the scheduling policy,
\begin{equation}\label{7}
Pr[t(i,j,k)=s]=\sum_{t=0}^s y_{i,j,k,t}\frac{p_{i,j,k,s-t}}{\mathbb{E}[S_{i,j,k}]}.
\end{equation}
From the above equation and constraints \eqref{3}-\eqref{4}, we have
\begin{eqnarray*}
&&\max_i\sum_j\sum_{k'\neq k}\mathbb{E}[S_{i,j,k'}]Pr[t(i,j,k')=s']\\
&=&\max_i\sum_j\sum_{k'\neq k}\sum_{t'=0}^{s'}y_{i,j,k',t'}p_{i,j,k',s'-t'}\leq 1;
\end{eqnarray*}
and
\begin{eqnarray*}
&&\max_j\sum_i\sum_{k'\neq k}\mathbb{E}[T_{i,j,k'}]Pr[t(i,j,k')=s']\\
&=&\max_j\sum_i\sum_{k'\neq k}\sum_{t'=0}^{s'}y_{i,j,k',t'}p_{i,j,k',s'-t'}\leq 1.
\end{eqnarray*}
Since we combine the set of flows with same tentative start time $s$ as a stochastic co-flow $\mathcal{J}(s)$, which has expected efficient size less or equal to $1$, the expected processing time of each $\mathcal{J}(s)$ is upper bounded by $H$ by assumption in this lemma. In the proposed algorithm, the flows with less tentative start time are scheduled earlier. Therefore, the summation of expected time scheduled before co-flow $(i,j,k)$ (suppose $t(i,j,k)=s$) is at most 
\begin{eqnarray*}
&&H\sum_{(i',j',k')\neq(i,j,k)}\mathbb{E}[S_{i',j',k'}]\Big{(}\sum_{s'=0}^{s-1}Pr[t(i',j',k')=s']\\
&+&\frac{1}{2}Pr[t(i',j',k')=s]\Big{)}\leq (s+\frac{1}{2})H.
\end{eqnarray*}
This proves the required result. 

\end{proof}

\cref{lma5.3} shows that the expected start time of any flow $(i,j,k)$ is at most $(t(i,j,k)+1/2)H$. \cref{cor}, given next, gives an upper bound of expected completion time of flow $(i,j,k)$. 
\begin{corollary}\label{cor}
Assume the expected time to process a stochastic co-flow with expected efficient size less or equal to $1$ by \cref{algo:GLJD} has upper bound $H$. For any $i,j\in\mathcal{M}$, $k\in\{1,\cdots,N\}$, and $s\in\{0,1,\cdots\}$, given $t(i,j,k)=s$, the conditional upper bound of expected completion time of flow $(i,j,k)$, $\mathbb{E}[C_{i,j,k}|t(i,j,k)=s]$ is at most $(s+1/2)H+\mathbb{E}[S_{i,j,k}]$. 
\end{corollary}
\begin{proof}
	The result follows directly from Lemma \ref{lma5.3} since the expected completion time can be obtained by adding the expected start time and the transfer time.
	\end{proof}

\begin{proposition}[Lemma 2, \cite{skutella2016unrelated}] \label{prop5.4}
For every $i,j\in\mathcal{M}$, $k\in\{1,\cdots,N\}$ and $r\in\{0,1,\cdots\}$, we have
\begin{equation*}
\sum_{r\in\mathbb{Z}_{\geq 0}}(r+\frac{1}{2})\frac{p_{i,j,k,r}}{\mathbb{E}[S_{i,j,k}]}=\frac{1+\mathbb{CV}[S_{i,j,k}]^2}{2}\mathbb{E}[S_{i,j,k}],
\end{equation*}
where $\mathbb{CV}[S_{i,j,k}]^2$ is $S_{i,j,k}$'s squared coefficient of variation.
\end{proposition}

The following few results will be needed to prove \cref{var} which gives an upper bound for $H$ given in \cref{lma5.3}. 

The next proposition is from Theorem 6 in \cite{keslassy2003guaranteed}, giving us a guarantee for GLJD decomposition

\begin{proposition}[Theorem 6, \cite{keslassy2003guaranteed}]\label{propgljd}
For any $s\in\{0,1,\cdots\}$, if a co-flow $\mathcal{J}(s)$ to be processed on $m$ servers has an efficient size less or equal to $1$, GLJD provides a decomposition $\mathcal{J}(s)=\sqcup_{l=1}^{l_s} {X^l}$ such that $\sum_{l\in\{1,\cdots,l_s\}}\max_{(i,j,k)\in X^l}\mathbb{E}[S_{i,j,k}]\leq (2\log{m}+1).$ 
\end{proposition}

The next corollary, Corollary \ref{prop5.5}, follows directly from the Proposition \ref{propgljd}.

\begin{corollary} \label{prop5.5}
Suppose the number of servers $|\mathcal{M}|=m\geq 2$. For any $s\in\{0,1,\cdots\}$, $\mathcal{J}(s)=\sqcup_{l=1}^{l_s}{{\bf X}(s)^l}$, we have $\sum_{l=1}^{l_s}\max_{(i,j,k)\in {\bf X}(s)^l}\mathbb{E}[S_{i,j,k}]\leq (2\log{m}+1)$. 
\end{corollary}

Recall that in Step 3 of \cref{algo:npscs}, we sum up all flows between same servers with same tentative start time, and regard the summed flow as a single flow. The following \cref{var} shows that the coefficient of variations of the resulting flow will be bounded by the one of the original flows.

\begin{lemma}\label{var}
Suppose $T_1,T_2,\cdots,T_k$ are $k$ independent random variables with expectation $\mathbb{E}[T_i]$ $(i\in\{1,\cdots,k\})$, and variances $Var(T_i)$ $(i\in\{1,\cdots,k\})$. Suppose the upper bound of their coefficient of variation is $\Delta$, then the random variable $T_1+T_2+\cdots+T_k$ has coefficient of variation that is upper bounded by $\Delta$.
\end{lemma}
\begin{proof}
We have 
\begin{equation*}
\mathbb{E}[T_1+T_2+\cdots+T_k]=\sum_{i=1}^k\mathbb{E}[T_i],
\end{equation*}
and
\begin{equation*}
Var(T_1+T_2+\cdots+T_k)=\sum_{i=1}^kVar(T_i),
\end{equation*}
since $T_1,T_2,\cdots,T_k$ are independent.

Recall that the coefficient of variation of a random variable $T$ is $Var(S)/\mathbb{T}[S]$. Note that $\Delta=\max_{i\in \{1,\cdots,k\}}\sqrt{Var(T_i)}/\mathbb{E}[T_i]$. Therefore,
\begin{eqnarray*}
&&\frac{\sqrt{Var(T_1+T_2+\cdots+T_k)}}{\mathbb{E}[T_1+T_2+\cdots+T_k]}\\
&=&\frac{\sqrt{Var(T_1)+\cdots+Var(T_k)}}{\mathbb{E}[T_1]+\cdots+\mathbb{E}[T_k]}\\
&\leq & \frac{\Delta\sqrt{\mathbb{E}[T_1]^2+\cdots+\mathbb{E}[T_k]^2}}{\mathbb{E}[T_1]+\cdots+\mathbb{E}[T_k]}\\
&\leq & \Delta.
\end{eqnarray*}
The first equality follows from the independence between $T_1,\cdots,T_k$, the first inequality follows from Jensen's inequality.
\end{proof}


Corollary \ref{prop5.5} provides a upper bound of $\sum_{l=1}^{l_s}\max_{(i,j,k)\in {\bf X}(s)^l}\mathbb{E}[S_{i,j,k}]$. However, we are more interested in $\sum_{l=1}^{l_s}\mathbb{E}[\max_{(i,j,S_{i,j,D(s)})\in {\bf X}(s)^l}S_{i,j,D(s)}]$ which is the expected total processing time of all co-flows ${\bf X}(s)^l$ ($l\in\{1,\cdots,l_s\}$) from $\mathcal{J}(s)$. The following \cref{prop5.6} connects the two expressions.
\begin{proposition}[\cite{devroye1979inequalities}] \label{prop5.6}
If $T_1,T_2,\cdots,T_m$ are $m$ random variables with finite means and finite variances, then 
\begin{equation*}
\mathbb{E}[\max_{i} T_i]\leq \max_i \mathbb{E}[T_i]+\sqrt{m}\max_i\sqrt{Var(T_i)},
\end{equation*}
\end{proposition}

In the following Lemma, we give a upper bound for $H$ in \cref{lma5.3}.

\begin{lemma} \label{lma5.7}
For every $s\in\{0,1,\cdots\}$, the expected total processing time of $\mathcal{J}(s)$ has upper bound $(2\log{m}+1)(1+\sqrt{m}\Delta)$, where $\Delta$ is the upper bound of squared coefficient of variation of all processing times, and $m$ is the number of servers. Namely, the upper bound on the expected time to process a stochastic co-flow with expected efficient size less or equal to $1$ is given as
\begin{equation*}
H= (2\log{m}+1)(1+\sqrt{m}\Delta).
\end{equation*}
\end{lemma}
\begin{proof}
The set of stochastic flows $\mathcal{J}(s)$ for any $s\in\{0,1,\cdots\}$, is a stochastic co-flow with efficient size less or equal to one:
\begin{equation*}
\max_{(i,j,k)\in\sigma(s)}\Big{\{}\max_i\sum_j \mathbb{E}[S_{i,j,k}],\max_j\sum_i \mathbb{E}[S_{i,j,k}]\Big{\}}\leq 1.
\end{equation*}

We use their representative matrice ${\bf D}(s)$ ($\forall s\in\{0,1,\cdots\}$) for GLJD algorithm. Suppose we get $l_s$ co-flows ${\bf X}(s)^1,\cdots,{\bf X}(s)^{l_s}$ with perfect matchings between the servers. From \cref{prop5.5}, 
\begin{equation}\label{f1}
\sum_{l=1}^{l_s}\max_{(i,j,k) \in {\bf X}(s)^l}\mathbb{E}[S_{i,j,k}]\leq 2\log m+1.
\end{equation}

However, the ultimate processing time for co-flow $\mathcal{J}(s)$ is $\mathbb{E}[\max_{(i,j,k) \in {\bf X}(s)^l}S_{i,j,k}]$. From Proposition \ref{prop5.6}, the expected total processing time for all co-flows $\{{\bf X}(s)^1,\cdots,{\bf X}(s)^{l_s}\}$ is
\begin{equation}\label{f2}
\sum_{l=1}^{l_s}\mathbb{E}[\max_{(i,j,k) \in X^l}S_{i,j,k}]\leq \sum_{l=1}^{l_s}\max_{(i,j,k) \in X^l}\mathbb{E}[S_{i,j,k}](\big{(}1+\sqrt{m}\Delta).
\end{equation}

Combining the inequalities \eqref{f1}-\eqref{f2}, the expected processing time of stochastic co-flow $\mathcal{J}(s)$ is at most
\begin{eqnarray*}
&&\sum_{l=1}^{l_s}\mathbb{E}[\max_{(i,j,k) \in {\bf X}(s)^l}S_{i,j,k}]\\
&\leq & \sum_{l=1}^{l_s}\max_{(i,j,k) \in{\bf X}(s)^l}\mathbb{E}[S_{i,j,k}](\big{(}1+\sqrt{m}\Delta)\\
&\leq & (2\log m+1)(1+\sqrt{m}\Delta),
\end{eqnarray*}
which proves the result as in the statement of the Lemma.
\end{proof}

Having the required results, we will now prove \cref{thm5.1}.
\begin{proof}
Recall that $y_{i,j,k,t}$ is the probability of starting flow $(i,j,k)$ at time slot $t$. The expected completion time of flow $(i,j,k)$: $\mathbb{E}[C_{i,j,k}]$ is $\sum_{t\in\mathbb{Z}_{\geq0}} y_{i,j,k,t}(t+\mathbb{E}[S_{i,j,k}])$.

\begin{eqnarray*}
&&\mathbb{E}[C_{i,j,k}]\\
&\overset{(a)}{=}&\sum_{s\in\mathbb{Z}_{\geq 0}}\mathbb{E}[C_{i,j,k}|t(i,j,k)=s]Pr[t(i,j,k)=s]\\
&\overset{(b)}{\leq}&\sum_{s\in\mathbb{Z}_{\geq 0}}\Big{(}H(s+\frac{1}{2})+\mathbb{E}[S_{i,j,k}]\Big{)}Pr[t(i,j,k)=s]\\
&\overset{(c)}{\leq}&\sum_{s\in\mathbb{Z}_{\geq 0}}\Big{(}H(s+\frac{1}{2})+\mathbb{E}[S_{i,j,k}]\Big{)}\sum_{t=0}^s y_{i,j.k.t}\frac{p_{i,j,k,s-t}}{\mathbb{E}[S_{i,j,k}]}\\
&\overset{(d)}{\leq}&\sum_{r\in\mathbb{Z}_{\geq 0}}\sum_{t\in\mathbb{Z}_{\geq 0}} \Big{(}H(t+r+\frac{1}{2})+\mathbb{E}[S_{i,j,k}]\Big{)}y_{i,j,k,t}\frac{p_{i,j,k,r}}{\mathbb{E}[S_{i,j,k}]}\\
&\overset{(e)}{\leq}& H\sum_{r\in\mathbb{Z}_{\geq 0}}\sum_{t\in\mathbb{Z}_{\geq 0}} \Big{(}t+r+\frac{1}{2}+\mathbb{E}[S_{i,j,k}]\Big{)}y_{i,j,k,t}\frac{p_{i,j,k,r}}{\mathbb{E}[S_{i,j,k}]}\\
&\overset{(f)}{\leq}& H\sum_{t\in\mathbb{Z}_{\geq 0}}y_{i,j,k,t}\Big{(}t+\mathbb{E}[S_{i,j,k}]+\sum_{r\in\mathbb{Z}_{\geq 0}}(r+\frac{1}{2})\frac{p_{i,j,k,r}}{\mathbb{E}[S_{i,j,k}]}\Big{)}\\
&\overset{(g)}{=}& H\sum_{t\in\mathbb{Z}_{\geq 0}}y_{i,j,k,t}\Big{(}t+\mathbb{E}[S_{i,j,k}]+\frac{1+{\mathbb{CV}[S_{i,j,k}]^2}}{2}\mathbb{E}[S_{i,j,k}]\Big{)}\\
&\overset{(h)}{=}& H\sum_{t\in\mathbb{Z}_{\geq 0}}y_{i,j,k,t}\Big{(}t+\frac{3+{\mathbb{CV}[S_{i,j,k}]^2}}{2}\mathbb{E}[S_{i,j,k}]\Big{)}\\
&\overset{(i)}{\leq}& H\sum_{t\in\mathbb{Z}_{\geq 0}}y_{i,j,k,t}\frac{3+{\mathbb{CV}[S_{i,j,k}]^2}}{2}\big{(}t+\mathbb{E}[S_{i,j,k}]\big{)}\\
&\overset{(j)}{\leq}& H\sum_{t\in\mathbb{Z}_{\geq 0}}y_{i,j,k,t}\frac{3+\Delta}{2}\big{(}t+\mathbb{E}[S_{i,j,k}]\big{)}\\
&\overset{(k)}{\leq}& H\frac{3+\Delta}{2}C_{i,j,k}^{LP}\\
&\overset{(l)}{\leq}& (2\log{m}+1)(1+\sqrt{m}\Delta)\frac{3+\Delta}{2}C_{i,j,k}^{LP},
\end{eqnarray*}
where $\Delta$ is the upper bound of $\mathbb{CV}[S_{i,j,k}]^2$ for all $i,j\in\mathcal{M}$, $k\in\{1,\cdots,N\}$. The above steps hold because of the following. $(a)$ is uncondtioning expectation
\begin{equation*}
\mathbb{E}[X]=\sum_{y}\mathbb{E}[X|Y=y]Pr[Y=y],
\end{equation*}
 $(b)$ follows from Corollary \ref{cor}, $(c)$ follows from \eqref{7}, $(d)$ sets $r=s-t$, $(e)$ extracts $H$ outside since $H\geq 1$, $(f)$ exchanges the summation order of $s$ and $t$, $(g)$ follows from \cref{prop5.4}, $(h)$ combines the two terms with $\mathbb{E}[S_{i,j,k}]$ together, $(i)$ extracts $\frac{3+{\mathbb{CV}[S_{i,j,k}]^2}}{2}\ge 1$ out, $(j)$ follows from the notation that $\Delta$ is the upper bound of all squared coefficient of variation of all variables $S_{i,j,k}$ $\forall i,j\in\mathcal{M},k\in\{1,\cdots,N\}$, $(k)$ follows from \eqref{5}, and $(l)$ follows from Lemma \ref{lma5.7}.

Since $\mathbb{E}[C_k]=\mathbb{E}[\max_{i,j\in[m]}C_{(i,j,k)}]$, note that $\{(i,j):i,j\in\{1,\cdots,m\}\}$ contains $m^2$ elements, applying Proposition \ref{prop5.6}, we have
\begin{equation*}
\mathbb{E}[C_k]\leq C_k^{LP}(2\log{m}+1)(1+\sqrt{m}\Delta)(1+m{\Delta}){(3+\Delta)}/{2}.
\end{equation*}
This proves the result as in the statement of Theorem \ref{thm5.1}.
\end{proof}

\begin{rem}
	We note that when the flow sizes are deterministic, the result in the statement of Theorem \ref{thm5.1} can be used with $\Delta=0$. 
\end{rem}

\section{Results for general release times}
\label{sec:release}

So far, we assumed that the co-flow tasks were released at time zero. If flow $(i,j,k)$ has release time $r_{i,j,k}$, it has zero possibility to be processed before $r_{i,j,k}$. We can simply add a set of constraints to LP problem \eqref{1}-\eqref{6}:
\begin{align*}
y_{i,j,k,t}=0 \qquad\qquad \forall i,j\in\mathcal{M},k\in\{1,\cdots,N\},t<r_{i,j,k}.
\end{align*}
We propose the same LP based algorithm, NPSCS, for this general case with the above additional constraint in the linear program. With this modification, the following result gives the approximation result for general release times.

\begin{theorem}
With release times constraints, the completion time of the $k$-th coflow tasks under scheduling algorithm NPSCS is at most $C_k^{LP}(2\log{m}+1)(1+\sqrt{m}\Delta)(1+m{\Delta})(2+\Delta)$, where $\Delta$ is the upper bound of $\mathbb{CV}[S_{i,j,k}]^2$ for all $i,j\in \mathcal{M}$, $k\in\{1,\cdots,N\}$. $\mathbb{CV}[S_{i,j,k}]^2 \triangleq (\mathbb{E}[S_{i,j,k}^2]-\mathbb{E}[S_{i,j,k}]^2)/\mathbb{E}[S_{i,j,k}]^2$ is the squared coefficient of variation of $S_{i,j,k}$.\label{genthm}
\end{theorem}

\begin{proof}
If a flow $(i,j,k)$ has release time more than $s$, then $y_{i,j,k,t}=0$ for all $t\leq s$. The tentative time generated cannot be less or equal to $s$. In other word, a flow with tentative start time $s$ has release time less or equal to $s$. The expected completion time of a flow with tentative start time $s$ is less than the summation of its release time (less than $s$) and the expected total processing time of flows before ($(s+1/2)H$) by Lemma \ref{lma5.3}.

Since $H\geq 1$, we have 
\begin{eqnarray}\label{8}
&&\mathbb{E}[C_{i,j,k}|t(i,j,k)=s]\nonumber\\
&\leq & s+(s+\frac{1}{2})H+\mathbb{E}[S_{i,j,k}]\nonumber\\
&=&(2s+\frac{1}{2})H+\mathbb{E}[S_{i,j,k}],
\end{eqnarray}

Further, we have
\begin{eqnarray*}
&&\mathbb{E}[C_{(i,j,k)}]\\
&\overset{(a)}{=}&\sum_{s\in\mathbb{Z}_{\geq 0}}\mathbb{E}[C_{(i,j,k)}|t(i,j,k)=s]Pr[t(i,j,k)=s]\\
&\overset{(b)}{\leq}&\sum_{s\in\mathbb{Z}_{\geq 0}}\Big{(}H(2s+\frac{1}{2})+\mathbb{E}[S_{i,j,k}]\Big{)}Pr[t(i,j,k)=s]\\
&\overset{(c)}{\leq}&\sum_{s\in\mathbb{Z}_{\geq 0}}\Big{(}H(2s+\frac{1}{2})+\mathbb{E}[S_{i,j,k}]\Big{)}\sum_{t=0}^s y_{i,j,k,t}\frac{p_{i,j,k,s-t}}{\mathbb{E}[S_{i,j,k}]}\\
&\overset{(d)}{\leq}&\sum_{r\in\mathbb{Z}_{\geq 0}}\sum_{t\in\mathbb{Z}_{\geq 0}} \Big{(}H(2t+2r+\frac{1}{2})+\mathbb{E}[S_{i,j,k}]\Big{)}y_{i,j,k,t}\frac{p_{i,j,k,r}}{\mathbb{E}[S_{i,j,k}]}\\
&\overset{(e)}{\leq}&H\sum_{r\in\mathbb{Z}_{\geq 0}}\sum_{t\in\mathbb{Z}_{\geq 0}} \Big{(}2t+2r+\frac{1}{2}+\mathbb{E}[S_{i,j,k}]\Big{)}y_{i,j,k,t}\frac{p_{i,j,k,r}}{\mathbb{E}[S_{i,j,k}]}\\
&\overset{(f)}{\leq}& H\sum_{t\in\mathbb{Z}_{\geq 0}}y_{i,j,k,t}\Big{(}2t+\mathbb{E}[S_{i,j,k}]+\sum_{r\in\mathbb{Z}_{\geq 0}}(2r+\frac{1}{2})\frac{p_{i,j,k,r}}{\mathbb{E}[S_{i,j,k}]}\Big{)}\\
&\overset{(g)}{\leq}& H\sum_{t\in\mathbb{Z}_{\geq 0}}y_{i,j,k,t}\Big{(}2t+\mathbb{E}[S_{i,j,k}]+2\sum_{r\in\mathbb{Z}_{\geq 0}}(r+\frac{1}{2})\frac{p_{i,j,k,r}}{\mathbb{E}[S_{i,j,k}]}\Big{)}\\
&\overset{(h)}{=}& H\sum_{t\in\mathbb{Z}_{\geq 0}}y_{i,j,k,t}\Big{(}2t+\mathbb{E}[S_{i,j,k}]+\big{(}1+{\mathbb{CV}[S_{i,j,k}]^2}\big{)}\mathbb{E}[S_{i,j,k}]\Big{)}\\
&\overset{(i)}{=}& H\sum_{t\in\mathbb{Z}_{\geq 0}}y_{i,j,k,t}\Big{(}2t+\big{(}2+{\mathbb{CV}[S_{i,j,k}]^2}\big{)}\mathbb{E}[S_{i,j,k}]\Big{)}\\
&\overset{(j)}{=}& H\sum_{t\in\mathbb{Z}_{\geq 0}}y_{i,j,k,t}\big{(}2+{\mathbb{CV}[S_{i,j,k}]^2}\big{)}\Big{(}t+\mathbb{E}[S_{i,j,k}]\Big{)}\\
&\overset{(k)}{\leq}& H\sum_{t\in\mathbb{Z}_{\geq 0}}y_{i,j,k,t}\big{(}2+{\Delta}\big{)}\big{(}t+\mathbb{E}[S_{i,j,k}]\big{)}\\
&\overset{(l)}{\leq}& H({2+\Delta})C_{(i,j,k)}^{LP}\\
&\overset{(m)}{\leq}& (2\log{m}+1)(1+\sqrt{m}\Delta)({2+\Delta})C_{(i,j,k)}^{LP},
\end{eqnarray*}
where $\Delta$ is the upper bound of $\Delta_{ij}$ for all $i,j\in [m]$, $m$ is the number of servers. The steps above can be explained as follows. $(a)$ is uncondtioning expectation
\begin{equation*}
\mathbb{E}[X]=\sum_{y}\mathbb{E}[X|Y=y]Pr[Y=y],
\end{equation*}
 $(b)$ follows from \eqref{8}, $(c)$ follows from \eqref{7}, $(d)$ sets $r=s-t$, $(e)$ extracts $H$ out since $H\geq 1$, $(f)$ exchanges the summation order of $s$ and $t$, $(g)$ extract $2$ out, $(h)$ follows from Proposition \ref{prop5.4}, $(i)$ combines the two terms having $\mathbb{E}[S_{i,j,k}]$, $(j)$ follows from the notation that $\Delta$ is the upper bound of all squared coefficient of all variables $S_{i,j,k}$ $\forall i,j\in\mathcal{M},k\in\{1,\cdots,N\}$, $(k)$ follows from \eqref{5}, and $(l)$ follows from Lemma \ref{lma5.7}.

Since $\mathbb{E}[C_k]=\mathbb{E}[\max_{i,j\in[m]}C_{(i,j,k)}]$, applying Proposition \ref{prop5.6}, we have
\begin{equation*}
\mathbb{E}[C_k]\leq C_k^{LP}(2\log{m}+1)(1+\sqrt{m}\Delta)(1+m{\Delta}){(2+\Delta)}.
\end{equation*}
This proves the result as in the statement of the Theorem.
\end{proof}

We also note that Theorem \ref{lma0} can also be easily extended with general release times, by changing 
\begin{equation}
F_1\triangleq mN\Big{(}\max_{k\in\{1,\cdots,N\}, i\in {\cal M}, j \in {\cal M}}r_{i,j,k}+\sum_{i\in\mathcal{M}}\sum_{j\in\mathcal{M}}\sum_{k=1}^N\mathbb{E}[S_{i,j,k}]\Big{)}.
\end{equation}
Thus, the number of time slots can be truncated, yielding a polynomial time algorithm.

\begin{rem}
	We note that when the flow sizes are deterministic, the result in the statement of Theorem \ref{genthm} can be used with $\Delta=0$. 
\end{rem}

\section{Conclusions}
\label{sec:con}
This paper studies stochastic non-preemptive co-flow scheduling, and gives an approximation algorithm. The results are provided for both zero and general release times. The results can also be specialized to deterministic co-flow scheduling.

\bibliographystyle{IEEEtran}
\bibliography{refs}

\begin{thebibliography}{10}
\providecommand{\url}[1]{#1}
\csname url@samestyle\endcsname
\providecommand{\newblock}{\relax}
\providecommand{\bibinfo}[2]{#2}
\providecommand{\BIBentrySTDinterwordspacing}{\spaceskip=0pt\relax}
\providecommand{\BIBentryALTinterwordstretchfactor}{4}
\providecommand{\BIBentryALTinterwordspacing}{\spaceskip=\fontdimen2\font plus
\BIBentryALTinterwordstretchfactor\fontdimen3\font minus
  \fontdimen4\font\relax}
\providecommand{\BIBforeignlanguage}[2]{{%
\expandafter\ifx\csname l@#1\endcsname\relax
\typeout{** WARNING: IEEEtran.bst: No hyphenation pattern has been}%
\typeout{** loaded for the language `#1'. Using the pattern for}%
\typeout{** the default language instead.}%
\else
\language=\csname l@#1\endcsname
\fi
#2}}
\providecommand{\BIBdecl}{\relax}
\BIBdecl

\bibitem{rmao2018infwrk}
R.~Mao, V.~Aggarwal, and M.~Chiang, ``Stochastic non-preemptive co-flow
  scheduling with time-indexed relaxation,'' in \emph{IEEE Infocom Worshop on
  Big Data in Cloud Performance (DCPerf)}, Apr 2018.

\bibitem{Dean:2008:MSD:1327452.1327492}
\BIBentryALTinterwordspacing
J.~Dean and S.~Ghemawat, ``Mapreduce: Simplified data processing on large
  clusters,'' \emph{Commun. ACM}, vol.~51, no.~1, pp. 107--113, Jan. 2008.
  [Online]. Available: \url{http://doi.acm.org/10.1145/1327452.1327492}
\BIBentrySTDinterwordspacing

\bibitem{shvachko2010hadoop}
K.~Shvachko, H.~Kuang, S.~Radia, and R.~Chansler, ``The hadoop distributed file
  system,'' in \emph{Mass storage systems and technologies (MSST), 2010 IEEE
  26th symposium on}.\hskip 1em plus 0.5em minus 0.4em\relax IEEE, 2010, pp.
  1--10.

\bibitem{spark}
M.~Zaharia, M.~Chowdhury, T.~Das, A.~Dave, J.~Ma, M.~McCauley, M.~J. Franklin,
  S.~Shenker, and I.~Stoica, ``Resilient distributed datasets: A fault-tolerant
  abstraction for in-memory cluster computing,'' in \emph{Proceedings of the
  9th USENIX conference on Networked Systems Design and Implementation}.\hskip
  1em plus 0.5em minus 0.4em\relax USENIX Association, 2012, pp. 2--2.

\bibitem{google}
\BIBentryALTinterwordspacing
{google}. Google dataflow. [Online]. Available:
  \url{https://www.google.com/events/io}
\BIBentrySTDinterwordspacing

\bibitem{chowdhury2011managing}
M.~Chowdhury, M.~Zaharia, J.~Ma, M.~I. Jordan, and I.~Stoica, ``Managing data
  transfers in computer clusters with orchestra,'' in \emph{ACM SIGCOMM
  Computer Communication Review}, vol.~41, no.~4.\hskip 1em plus 0.5em minus
  0.4em\relax ACM, 2011, pp. 98--109.

\bibitem{Chowdhury:2012:CNA:2390231.2390237}
\BIBentryALTinterwordspacing
M.~Chowdhury and I.~Stoica, ``Coflow: A networking abstraction for cluster
  applications,'' in \emph{Proceedings of the 11th ACM Workshop on Hot Topics
  in Networks}, ser. HotNets-XI.\hskip 1em plus 0.5em minus 0.4em\relax New
  York, NY, USA: ACM, 2012, pp. 31--36. [Online]. Available:
  \url{http://doi.acm.org/10.1145/2390231.2390237}
\BIBentrySTDinterwordspacing

\bibitem{lawler1989sequencing}
E.~L. Lawler, J.~K. Lenstra, A.~H.~R. Kan, and D.~B. Shmoys, \emph{Sequencing
  and scheduling: Algorithms and complexity}.\hskip 1em plus 0.5em minus
  0.4em\relax CWI. Department of Operations Research, Statistics, and System
  Theory [BS], 1989.

\bibitem{kavi2001scheduled}
K.~M. Kavi, R.~Giorgi, and J.~Arul, ``Scheduled dataflow: Execution paradigm,
  architecture, and performance evaluation,'' \emph{IEEE Transactions on
  Computers}, vol.~50, no.~8, pp. 834--846, 2001.

\bibitem{yu2016non}
R.~Yu, G.~Xue, X.~Zhang, and J.~Tang, ``Non-preemptive coflow scheduling and
  routing,'' in \emph{Global Communications Conference (GLOBECOM), 2016
  IEEE}.\hskip 1em plus 0.5em minus 0.4em\relax IEEE, 2016, pp. 1--6.

\bibitem{Qiu:2015:MTW:2755573.2755592}
\BIBentryALTinterwordspacing
Z.~Qiu, C.~Stein, and Y.~Zhong, ``Minimizing the total weighted completion time
  of coflows in datacenter networks,'' in \emph{Proceedings of the 27th ACM
  Symposium on Parallelism in Algorithms and Architectures}, ser. SPAA
  '15.\hskip 1em plus 0.5em minus 0.4em\relax New York, NY, USA: ACM, 2015, pp.
  294--303. [Online]. Available:
  \url{http://doi.acm.org/10.1145/2755573.2755592}
\BIBentrySTDinterwordspacing

\bibitem{shafiee2017scheduling}
M.~Shafiee and J.~Ghaderi, ``Scheduling coflows in datacenter networks:
  Improved bound for total weighted completion time,'' in \emph{Proceedings of
  the 2017 ACM SIGMETRICS/International Conference on Measurement and Modeling
  of Computer Systems}.\hskip 1em plus 0.5em minus 0.4em\relax ACM, 2017, pp.
  29--30.

\bibitem{DBLP:journals/corr/ImP17}
\BIBentryALTinterwordspacing
S.~Im and M.~Purohit, ``A tight approximation for co-flow scheduling for
  minimizing total weighted completion time,'' \emph{CoRR}, vol.
  abs/1707.04331, 2017. [Online]. Available:
  \url{http://arxiv.org/abs/1707.04331}
\BIBentrySTDinterwordspacing

\bibitem{rothkopf1966scheduling}
M.~H. Rothkopf, ``Scheduling with random service times,'' \emph{Management
  Science}, vol.~12, no.~9, pp. 707--713, 1966.

\bibitem{keslassy2003guaranteed}
I.~Keslassy, M.~Kodialam, T.~Lakshman, and D.~Stiliadis, ``On guaranteed smooth
  scheduling for input-queued switches,'' in \emph{INFOCOM 2003. Twenty-Second
  Annual Joint Conference of the IEEE Computer and Communications. IEEE
  Societies}, vol.~2.\hskip 1em plus 0.5em minus 0.4em\relax IEEE, 2003, pp.
  1384--1394.

\bibitem{Chowdhury:2014:ECS:2740070.2626315}
\BIBentryALTinterwordspacing
M.~Chowdhury, Y.~Zhong, and I.~Stoica, ``Efficient coflow scheduling with
  varys,'' \emph{SIGCOMM Comput. Commun. Rev.}, vol.~44, no.~4, pp. 443--454,
  Aug. 2014. [Online]. Available:
  \url{http://doi.acm.org/10.1145/2740070.2626315}
\BIBentrySTDinterwordspacing

\bibitem{mckeown1999achieving}
N.~McKeown, A.~Mekkittikul, V.~Anantharam, and J.~Walrand, ``Achieving 100\%
  throughput in an input-queued switch,'' \emph{IEEE Transactions on
  Communications}, vol.~47, no.~8, pp. 1260--1267, 1999.

\bibitem{mckeown1999islip}
N.~McKeown, ``The islip scheduling algorithm for input-queued switches,''
  \emph{IEEE/ACM Transactions On Networking}, vol.~7, no.~2, pp. 188--201,
  1999.

\bibitem{skutella2016unrelated}
M.~Skutella, M.~Sviridenko, and M.~Uetz, ``Unrelated machine scheduling with
  stochastic processing times,'' \emph{Mathematics of operations research},
  vol.~41, no.~3, pp. 851--864, 2016.

\bibitem{skutella1998approximation}
M.~Skutella, \emph{Approximation and randomization in scheduling}.\hskip 1em
  plus 0.5em minus 0.4em\relax Cuvillier, 1998.

\bibitem{devroye1979inequalities}
L.~P. Devroye, ``Inequalities for the completion times of stochastic pert
  networks,'' \emph{Mathematics of Operations Research}, vol.~4, no.~4, pp.
  441--447, 1979.

\end{thebibliography}

\end{document}